%% file: main.tex
\RequirePackage{xcolor}
\RequirePackage{amsmath}
\RequirePackage{amsfonts}
\documentclass{llncs}
\newcommand{\blarghQED}{\hfill$\square$}
\usepackage{smwifclass}
\usepackage[utf8]{inputenc}
\usepackage{algorithm}
\usepackage[endLComment=~]{algpseudocodex}

\usepackage{balance}
\definecolor{linkc}{rgb}{0.1,0.1,.8}
\definecolor{darkgreen}{rgb}{0,0.5,0}
\definecolor{midblue}{rgb}{0,0,0.7}
\usepackage[bookmarks=false,colorlinks=true,citecolor=linkc,linkcolor=linkc,filecolor=linkc,urlcolor=linkc]{hyperref}
\usepackage{url}

\usepackage{tikz}
\usetikzlibrary{patterns}
\usepackage{listings}
\newcommand{\lstFigSize}{\scriptsize}
\newcommand{\lstDisplaySize}{\footnotesize}
\lstdefinestyle{lstFigStyle}{
   mathescape=true,
   escapechar=@,
   basicstyle=\lstFigSize\ttfamily,
   keywordstyle=\lstFigSize\color{midblue}\ttfamily\bfseries,
   commentstyle=\lstFigSize\slshape\color{darkgreen},
   numbers=none,
}  
\lstdefinestyle{lstDisplayStyle}{
   xleftmargin=\parindent,
   mathescape=true,
   escapechar=@,
   basicstyle=\lstDisplaySize\ttfamily,
   keywordstyle=\lstDisplaySize\color{midblue}\ttfamily\bfseries,
   commentstyle=\lstDisplaySize\slshape\color{darkgreen},
   numbers=none
}  
\input{listings_maple_definitions.sty}
\usepackage{mathdots}
\usepackage{smwmath}
\usepackage{smwtools}

\newcommand{\code}[1]{\texttt{#1}}
\newcommand{\aprec}[1]{\ensuremath{\mathrm{prec}_{#1}\,}}
\newcommand{\gprec}{\ensuremath{\mathrm{prec}\,}}
\newcommand{\adeg}[1]{\ensuremath{\mathrm{deg}_{#1}\,}}
\newcommand{\gdeg}{\ensuremath{\mathrm{deg}\,}}
\newcommand{\aclip}[1]{\ensuremath{\mathrm{clip}_{#1}\,}}
\newcommand{\notationsep}{\vspace{1ex}}
\newcommand{\arev}[1]{\ensuremath{\mathrm{rev}_{#1}\,}}

\newcommand{\mtwo}[1]{\ensuremath{\left [ \begin{array}{cc}#1\end{array} \right ]}}
\begin{document}
\title{Computing Clipped Products}
\author{
  Arthur C. Norman$^1$ 
  \and 
  Stephen M. Watt$^2$
}
\institute{
  Trinity College,
  Cambridge CB2 1TQ, UK\\
  \email{acn1@cam.ed.uk}\\
  ~
  \and
  Cheriton School of Computer Science,
  University of Waterloo, N2L 3G1 Canada\\
  \url{https://cs.uwaterloo.ca/~smwatt} \\
  \email{smwatt@uwaterloo.ca}
}

\maketitle
\begin{abstract}
Sometimes only some digits of a numerical product or some terms of a polynomial or series product are required.   
Frequently these constitute the most significant or least significant part of the value, for example when computing initial values or refinement steps in iterative approximation schemes. 
Other situations require the middle portion.
In this paper we provide algorithms for the general problem of computing a given span of coefficients within a product, that is the terms within a range of degrees for univariate polynomials or range digits of an integer. This generalizes the ``middle product'' concept of Hanrot, Quercia and Zimmerman. We are primarily interested in problems of modest size where constant speed up factors can improve overall system performance, and therefore focus the discussion on classical and Karatsuba multiplication and how methods may be combined.
\end{abstract}
\keywords{integer product, polynomial product, convolution subrange,\newline
product approximation}
\thispagestyle{plain}\pagestyle{plain}

\section{Introduction}
\label{sec:introduction}
The classical discussion of integer multiplication concentrates on
forming the product of a pair of $N$-digit numbers to form a $2N$ digit
result. The techniques used apply also to dense univariate polynomials or
to truncated power series. We start by pointing out what a limited
viewpoint this represents.

There are three use-cases where at first it seems that 
balanced multiplication (\textit{i.e.} ones where the two inputs are the same size) will be central in large computations. The first
is in fixed but extended precision floating point arithmetic where $N$-digit
floating values are handled as $N$-digit integers alongside an exponent
value. But here it is not really appropriate to use full general multiplication because only the top $N$ of the
product digits are retained, so computing the low $N$ digits will represent
a waste. The same issue arises (but minus the complication of carries) in
work with truncated power series where it is the low half of the full polynomial product that has to be kept.
The second is in cryptography where there is some large modulus $p$ and all
values worked with are in the range 0 to $p-1$. But hereafter any
multiplication that generates a double-length product there needs to be a
reduction mod $p$. This has the feel of (even if it is not exactly the same
as) just wanting the low $N$ digits of the full $2N$ digit product.
Often the cost of the remaindering operation will be at least as important as the multiplication, and we will refer later to our thoughts on that.
The third may be in computation of elementary constants such as $e$ and $\pi$
to extreme precision where the precision of all the computational steps is very rigidly choreographed, but again at least at the
end it is liable to have much in common with the extended precision floating point case.

Other real-world cases are more liable to be forming $M\times N$ products where $M$
and $N$ may happen to be close in magnitude but can also be wildly
different, and where the values of $M$ and $N$ are liable to differ for
each multiplication encountered. Such a situation arises in almost any computation involving polynomials with exact integer coefficients --- GCD
calculations, Gr\"{o}bner bases, power series with rational 
coefficients, quantifier elimination. Many of these tasks have both broad
practical applicability and can be seriously computationally demanding so that absolute performance matters.
We expand on these arguments in~\cite{2024-norman-watt}.


Part of our motivation for this work was that, when considering division (and hence remainder) using a Newton-style iteration to compute a scaled inverse of the divisor, we found that in a fairly natural way we were wanting to multiply an $N$ digit value by one of length $2N$ but then only use the middle $N$ digits of the full $3N$-length product. 
Failing to exploit both the imbalance and the clipping there would hurt the overall performance of division.  

This specific problem of computing the middle third of a $2N \times N$ product has been considered previously~\cite{HanQuerZim}, where it has been used to speed up the division and square root of power series.
In this situation, the middle third is exactly the part of the product where the convolution for each term uses all terms of the shorter factor.   This allows a Karatsuba-like recursive scheme to compute the middle third of the product, or to use FFTs of size $2N$ instead of $4N$, saving a factor of 2 multiplying large values. 
This reduces the Newton iteration time for division from $3M(N)+O(N)$ to $2M(N)+O(N)$,
where $M(N)$ is the time for multiplication of size $N$.

Another important case is that of computing the initial terms of a product. Mulders' algorithm~\cite{mulders,HanZim} produces the first $N$ terms of a power series product, called a ``short product''. 
This is done by selecting a cut off point $k$ and computing one $k\times k$ product and two $(N-k)\times(N-k)$ short products recursively.

These ``middle third of a $2N\times N$ product'' and ``prefix'' situations are important cases but proper code needs to allow for additional generality.
So we view the proper general problem to address is
the formation of an $M\times N$ product where only digits (or terms) from positions $a$ to $b$
are required.
We call this a \emph{clipped product}.

A clipped product can obviously be implemented by forming the full $M\times N$ product, or
indeed by padding the smaller of the two inputs to get the $N\times N$ case (and
sometimes going further and padding the size there to be a power of 2!)
and then at the end just ignoring unwanted parts of the result. From
an asymptotic big-$O$ perspective that may suffice, but we are interested
in implementations where even modest improvements in absolute performance
on ``medium sized'' problems matters. So can we do better? Achieving
certified lower bounds on cost becomes infeasible both with the large
number of input size and size-related parameters and with the fact that
for medium sized problems overheads matter, and they can be platform- and
implementation-sensitive. But despite that we will show we can suggest
better approaches than the na\"{\i}ve one.

The main results of this paper are algorithms for clipped multiplication of integers and polynomials over a general ring, adapting both the classical $O(N^2)$ and Karatsuba $O(N^{log_2 3})$ methods. An analysis is given to show how much look-back is required in order to have the correct carry in for the lowest digit in the classical integer case.

The remainder of the article is organized as follows:
Section~\ref{sec:preliminaries} provides some notation and gives a definition of clipping.
Section~\ref{sec:straightforward} shows some straightforward methods to compute clipped products.   The main idea is that one can compute a lower part by multiplying only the lower terms and clipping afterwards.  For polynomials, one can use reversed polynomials to get a higher part. 
Section~\ref{sec:clippedpoly} shows how to adapt polynomial multiplication algorithms so that only the required  part is calculated.  
Section~\ref{sec:clippedint} shows how to adapt integer multiplication so that only the required part is calculated.  The difference from polynomial multiplication is in dealing with carries.
Section~\ref{sec:combining} discusses some issues that arise in combining sub-multiplications.
Finally, Section~\ref{sec:thoughts} presents some further thoughts and conclusions.  The algorithms are presented in high-level Maple code, though in practice a lower-level language like C or Rust may be used.

\section{Preliminaries}
\label{sec:preliminaries}
For many algorithms, integers and univariate polynomials behave similarly, however especially when losing low order digits is called for integer calculations are complicated by the need to allow for carries up from those discardable parts of the results.  
We use the notation of~\cite{watt-issac-2023,watt-casc-2023} to allow generic discussion of concepts relating to both:

\notationsep
\begin{tabular}{p{8em}l}
$[a..b], [a..b)$, \textit{etc} & integer intervals, \textit{i.e.} real intervals intersected with $\mathbb Z$ \\
$\aprec B n$  &  number of base-$B$ digits of an integer $n$, $\lfloor \log_B |n|\rfloor + 1$\\
$\aprec x p$  & number of coefficients of a polynomial $p$, $\adeg x p + 1$\\
$M(n,m)$ & the time to multiply two values with $\gprec$ $n$ and $m$. 
\end{tabular}
\notationsep

\noindent
The integer interval notation, ``$[a..b]$'' \textit{etc}, is used by Knuth, \textit{e.g.}~\cite{knuth-vol4b}.
In discussions that apply to both integers and polynomials we may use $t$ as a generic base, which may stand for a integer radix or a polynomial variable, in which case we may write, \textit{e.g.}, $\aprec t u$.
We write the coefficients of $u$ as $u_{t,i}$, where
\[
   u = \sum_{i\in[0..\aprec t u)} u_{t,i} t^i.
\]
In the integer case, it is required to have $0 \le u_{B,i} < B$ for uniqueness.  If $i < 0$ or $i \ge \aprec t u$, then $u_{t, i} = 0$.
When the base for integers or variable for polynomials is understood, then we may simply write $\gprec u$ and $u_i$.   
To refer to part of an integer or polynomial, we use the notation below and call this the \emph{clipped} value.

\notationsep
\begin{tabular}{p{8em}l}
$\aclip{t, I} u $ & $\sum_{i \in I} u_{t,i} t^i$, $I$ an integer interval,
\end{tabular}
\notationsep

\noindent
As before, when the base for integers or the variable for polynomials is understood, we may simply write $\aclip I u$.

\paragraph{Examples}~\\\noindent
Letting $p = a_5 x^5 + a_4 x^4 + a_3 x^3 + a_2 x^2 + a_1 x + a_0$,
\begin{align*}
    \aprec x p &= 6 \\
    \aclip{x,[2..4)} p &= a_3 x^3 + a_2 x^2\\
\end{align*}
Letting $n=504132231405$,
\begin{align*}
    \aprec {100} n &= 6 \\
    \aclip{100,[2..4]} n &= 4132230000.
\end{align*}

\section{Straightforward Methods}
\label{sec:straightforward}
We begin with the obvious methods to compute the clipped value of a product.
\subsection*{Direct Clipped Product}
The simplest approach is simply to compute the whole product and extract the desired part. The clipped product $\aclip{t, [a..b]}(f \times g)$ may be computed as shown in Figure~\ref{fig:straightforward_clipped}.
Here, \verb+clip(p, t, a, b)+ computes $\aclip{t,[a..b]} p$.
The computation time is $M(\gprec f, \gprec g) + O(|r|)$.
\paragraph{Example}
With
\begin{align*}
f &= 4x^3+83x^2+10x-62, &
g &= 82x^5-80x^4+44x^3-71x^2+17x+75
\end{align*}
we have 
\[
f \times g =
328x^8+6486x^7-5644x^6-2516x^5-425x^4-1727x^3+10797x^2-304x-4650
\]
so
\[
\aclip{x,[2..3]}(f\times g) =
-1727x^3+10797x^2.
\]

\subsection*{Bottom Clipped Product}
It is possible to improve on the direct clipped product when only the lower index coefficients are required. In this case integers and polynomials can be treated identically since carries out from the calculated range have no impact on the final result.   If the clipped product range is $[0..b]$ then no coefficient of index greater than $b$ of the operands can affect the clipped product.
Therefore those coefficients can be discarded prior to computing the product.  Then, as shown in Figure~\ref{fig:straightforward_clipped}, one proceeds as before.
The computation time is \[
 M_{\mathrm{bot}} (\gprec f, \gprec g, b) = M(\min(\gprec f, b) , \min(\gprec g, b)) + O(b).
 \]
 \paragraph{Example}
 
 With $f$ and $g$ as before, we may compute $\aclip{x,[0..3]} (f\times g)$
 by multiplying
 $\aclip{x,[0..3]} f=4x^3 + 83x^2 + 10x - 62$ and
 $\aclip{x,[0..3]} g=44x^3 - 71x^2 + 17x + 75$ to get
 $
 176x^6 + 3368x^5 - 5385x^4 - 1727x^3 + 10797x^2 - 304x - 4650
 $ which is then clipped to retain the terms of degrees in $[0..3]$, namely
 $- 1727x^3 + 10797x^2 - 304x - 4650$.

\subsection*{General Clipped Product from Bottom}
If it is desired to compute a clipped product with a range with lower bound other than zero, it is possible to do so with one call to \code{BottomClippedProduct},
as shown in Figure~\ref{fig:straightforward_clipped}.
The computation time is 
\[
M_{\mathrm{gbot}}(\gprec f, \gprec g, a, b) = 
M(\min(\gprec f, b), \min(\gprec g, b)) + O(b-a).
\]

\paragraph{Example} To compute $\aclip{x,[2..3]}(f\times g)$ we obtain $\aclip{x,[0..3]}(f\times g)$ using the bottom clipped product, and clip it as $\aclip{x,[2..3]}$.

\input{fig_straightforward_clipped}
\subsection*{Top Clipped Product From Reverse}
As the upper limit of the clipping range approaches the precision of the multiplicands, \code{BottomClippedProduct} loses its computational advantage.  In particular, if the top part of the product is desired, then  there is no advantage at all.   While the previous methods apply equally to integers and polynomials, here carries complicate the integer case. It is possible to compute polynomial clipped products with higher ranges only using polynomial reversal.  Let
\[
\arev x p = x^{\adeg x p} p(1/x).
\]
Then the product clipped to $[a..\gdeg f + \gdeg g]$, \textit{i.e.} the top, may be computed as shown in Figure~\ref{fig:top_clipped_via_reverse}.
\begin{figure}[t]
\begin{lstlisting}[language=maple,style=lstFigStyle]
TopClippedPolynomialProduct := proc(f, g, x, a)
    local degf, degg, revf, revg, p, r;
    degf := degree(f, x); degg := degree(g, x);
    revf := rev(f, x);    revg := rev(g, x);
    p := BottomClippedProduct(revf, revg, x, degf+degg-a);
    r := x^a * rev(p, x);
    return r
end
\end{lstlisting}
\caption{Top clipped product using polynomial reversal}
\label{fig:top_clipped_via_reverse}
\end{figure}
The time is determined by the \code{BottomClippedProduct} computation,
\[
  M_{\mathrm{bot}}(\gprec f, \gprec g, \gprec f + \gprec g + a - 2).
\]
\paragraph{Example}
To compute $\aclip{x,[6..8]} (f \times g),$
we first compute 
\begin{align*}
\arev x f &= -62 x^3 +10x^2 +83x +4,
&
\arev x g &= 75 x^5 +17x^4 -71x^3 +44x^2 -80x + 82.
\end{align*}
Then $\adeg x f + \adeg x g - a = 3 + 5 - 6 = 2.$  So we compute
\[
p = \aclip {x,[0..2]}(\arev x f \times \arev x g) =-5644*x^2 + 6486*x + 328 
\]
and the result is
\[
x^6 \arev x p = 328x^8 + 6486x^7 -5644x^6.
\]

\subsection*{Disadvantage of Straightforward Methods}
These straightforward methods provide clipped products that generally cost less than clipping a full product, but they still perform significant extra work.   
Forming the $[a..b]$ clipped product via \code{BottomClippedProduct} on $0..b$ computes $2b+1$ coefficients where only $b-a+1$ are required.  This is a significant difference for $O(N^p)$ multiplication methods such as the classical or Karatsuba algorithms.
We therefore consider how to improve on this.

 Mulders~\cite{mulders} provides a more sophisticated method to compute products clipped to $[0..b]$ and hence the others.
 Mulders' method is discussed further in Section~\ref{sec:combining}.  
 If we do not require the entire prefix, that is for general $[a..b]$, then we can compute only what is needed, as shown in the next sections.

\section{Clipped Polynomial Products}
\label{sec:clippedpoly}
For the moment we concentrate on polynomial products, since this avoids the technicalities of dealing with carries.  In practice, different multiplication algorithms give best performance for ranges of problem size.   We therefore consider modified classical, Karatsuba and FFT multiplication.

\subsection*{Clipped Classical Polynomial Multiplication}
Classical multiplication of univariate polynomials $f$ and $g$ requires $\gprec f \times \gprec g$ coefficient multiplications and a similar number of coefficient additions.  It is easy to compute only the desired coefficients, as shown in Figure~\ref{fig:clipped_classical_poly}.
%
\begin{figure}[t]
\begin{lstlisting}[language=maple,style=lstFigStyle]
ClippedClassicalPolynomialProduct := proc(f0, g0, x, a, b)
    local f, g, df, dg, s, k, t, i, i0;
    
    # Ensure dg <= df
    f := f0; df := degree(f, x);
    g := g0; dg := degree(g, x);
    if dg > df then f,g := g,f;  df,dg := dg,df fi;

    # Form column sums.
    s := 0;                    # Note A
    for k from a to b do
        if   k <= dg then i0 := 0
        elif k <= df then i0 := k-dg
        else              i0 := k-df
        fi;
        t := 0;
        for i from i0 to k do t := t + coeff(f,x,i) * coeff(g,x,k-i) od;
        s := setcoeff(s,x,k,t) # Note A. s + t*x^k
    od;
    return s
end
\end{lstlisting}
\caption{Clipped classical polynomial multiplication}
\label{fig:clipped_classical_poly}
\end{figure}
This is Maple code, so \code{coeff(f,x,i)} computes $f_i$, the coefficient of $x^i$.  The statement \code{s := setcoeff(s,x,k,t)} uses an auxilliary function to set the coefficient of $x^k$ in $s$ to be $t$.
Depending on the implementation language, the lines annotated \code{Note A} would typically be setting values in a coefficient array indexed from zero.   A suitable representation would be as a pair of the integer value \code{a} and a coefficient array \code{coeffs}, where \code{coeffs[i]} was the coefficient of $x^{a+i}$.

The greatest number of coefficient operations occurs when $a$ and $b$ both are in $[\gdeg g..\gdeg f]$, in which case $(b-a+1)*(\gdeg g+1)$ multiplications and $(b-a+1)*(\gdeg g)$ additions are required.  We therefore have the bound
\[
M_{\mathrm{cclip}}(\gprec f, \gprec g, a, b) \le (b-a+1)\times (\min(\gdeg f, \gdeg g) + 1).
\]

\paragraph{Example}
Let $\gdeg f = 7$, $\gdeg g = 4$, $a = 5$ and $b = 7$.  Then the multiplication table looks like the following with the bold face entries calculated:

\begin{center}
\begin{tabular}{c@{~~}c@{~~}c@{~~}c@{~~}c@{~~}c@{~~}c@{~~}c@{~~}c@{~~}c@{~~}c@{~~}c}
        &        &        &        &$\mathbf{f_7g_0}$&$\mathbf{f_6g_0}$&$\mathbf{f_5g_0}$&$f_4g_0$&$f_3g_0$&$f_2g_0$&$f_1g_0$&$f_0g_0$  \\
        &        &        &$f_7g_1$&$\mathbf{f_6g_1}$&$\mathbf{f_5g_1}$&$\mathbf{f_4g_1}$&$f_3g_1$&$f_2g_1$&$f_1g_1$&$f_0g_1$&          \\
        &        &$f_7g_2$&$f_6g_2$&$\mathbf{f_5g_2}$&$\mathbf{f_4g_2}$&$\mathbf{f_3g_2}$&$f_2g_2$&$f_1g_2$&$f_0g_2$&        &          \\
        &$f_7g_3$&$f_6g_3$&$f_5g_3$&$\mathbf{f_4g_3}$&$\mathbf{f_3g_3}$&$\mathbf{f_2g_3}$&$f_1g_3$&$f_0g_3$&        &        &          \\
$f_7g_4$&$f_6g_4$&$f_5g_4$&$f_4g_4$&$\mathbf{f_3g_4}$&$\mathbf{f_2g_4}$&$\mathbf{f_1g_4}$&$f_0g_4$&        &        &        &     
\end{tabular}
\end{center}
So here 15 multiplications and 12 additions are required rather than the 40 and 28 that would be required by \code{ClippedProductFromBottom}.

\subsection*{Clipped Karatsuba Polynomial Multiplication}
The Karatsuba multiplication scheme~\cite{karatsuba} splits the multiplicands in two and forms the required product using three recursive multiplications and four additions.
For simplicity, assume $\gprec f = \gprec g = p = 2^n$.
Let 
$f = f_h x^{p/2} + f_l$,
$g = g_h x^{p/2} + g_l$,
$f_m = f_h + f_l$ and $g_m = g_h+ g_l$.
For Karatsuba multiplication the three products $f_h \cdot g_h$,  $f_m \cdot g_m$ and $f_l \cdot g_l$ are required.
Then the product $f \cdot g = z_h x^p + z_m x^{p/2} + z_l$ where
\begin{align*}
    z_h &= f_h g_h, &
    z_m &= f_m g_m - f_h g_h - f_l g_l, &
    z_l &= f_l g_l.
\end{align*}

We show that using clipping while computing a product by Karatsuba multiplication can significantly reduce the cost compared to clipping after the product is computed.
Let $K(p)$ denote the number of required coefficient multiplications for Karatsuba multiplication of polynomials of $\gprec = p$. We have
$K(p) = 3 K(p/2)$ for $p > 1$ so $K(p) = p^{\log_2 3}$.
If only a top (or bottom) part of the product is needed, then less work is required.  
Suppose only a top portion of the coefficients are required.
Instead of computing all the required smaller products in full, some may be ignored and some may require only their upper parts.  
This reasoning leads to the following result.

\begin{theorem}
When multiplying polynomials of degree $p-1$ by Karatsuba's method, 
if only the top (or bottom)
  $1/2^\ell$ 
fraction of the coefficients are required, 
with $2^\ell < p$,
then the number of required coefficient multiplications
is at most $K(p)/3^{\ell-1}$.
\label{thm:topfrac}
\end{theorem}
\begin{proof}
Let $f$ and $g$ be the two precision $p$ polynomials to be multiplied and let $T(p,\ell)$ denote the number of coefficient multiplications required for the top $1/2^\ell$ fraction of the product.
If $\ell = 1$,
then the full product $z_h$ is needed, as is the top half of $z_m$.  This requires all of $f_h g_h$ and the top half of $f_m g_m$ and $f_l g_l$, so
\begin{align*}
    T(p,1) &= K(p/2) + 2T(p/2, 1) \\
           &\le K(p).
\end{align*}
If $\ell > 1$, then only the top $1/2^{\ell -1}$ fraction of the coefficients of $z_h$ are required and neither $z_m $ nor $z_l$ are needed. So
\begin{align*}
    T(p, \ell) &= T(p/2, \ell-1) \\
               &\le K(p/2^{\ell-1})  = K(p) / 3^{\ell-1}.
\end{align*}
Together, these cases give the result for the top $1/2^\ell$ fraction.  A similar argument gives the result for the bottom.
\blarghQED
\end{proof}
\noindent
Note that for $p = 2^N$, we have
\(
    K(p) / 3^{\ell-1} = 3^{N-\ell+1}+1, 
\)
an integer.
\\[\baselineskip]
The idea of pushing the clipping range down onto the Karatsuba sub-products is straightforward to implement, as shown in Figure~\ref{fig:clipped_karatsuba_poly}.
\begin{figure}[t]
\begin{lstlisting}[language=maple,style=lstFigStyle]
ClippedKaratsubaPolynomialProduct := proc(f, g, x, a, b, gmul)
    local df, dg, p, fh,gh, fl,gl,fm, gm, restrict,
          zha, zhb, zma, zmb, zla, zlb, zh, zm, zl;
          
    df := gdegree(f, x);  dg := gdegree(g, x);@\verb+                       +@# Note A

    if a > df+dg then return 0 fi;
    if a > b     then return 0 fi;
    if b = 0     then return gmul(coeff(f,x,0),coeff(g,x,0)) fi; # Note A

    # Size of z parts.
    p := max(df, dg)+1;
    p := p + irem(p, 2);    # Make even

    # Product factors.
    fh := clip(f,x,p/2,p-1)*x^(-p/2); fl := clip(f,x,0,p/2-1);@\verb+    +@# Note A
    gh := clip(g,x,p/2,p-1)*x^(-p/2); gl := clip(g,x,0,p/2-1);@\verb+    +@# Note A
    fm := fh + fl; gm := gh + gl;
    
    # Cases where zm is not needed.
    if b < p/2 then
        return ClippedKaratsubaPolynomialProduct(fl,gl,x,a,  b,  gmul)
    elif a > 3*p/2-2 then
        return ClippedKaratsubaPolynomialProduct(fh,gh,x,a-p,b-p,gmul)
               * x^p
    fi;

    # Need all products. 
    restrict :=   # return i or the closest interval endpoint.
        proc(i) if i < 0 then 0 elif i > p-2 then p-2 else i fi end:
    zha := restrict(a - p);   zhb := restrict(b - p);
    zma := restrict(a - p/2); zmb := restrict(b - p/2);
    zla := restrict(a);       zlb := restrict(b);

    # Expand high and low product ranges as needed for zm.
    zha := min(zha, zma); zhb := max(zhb, zmb);
    zla := min(zla, zma); zlb := max(zlb, zmb);
    zh  := ClippedKaratsubaPolynomialProduct(fh,gh,x,zha,zhb,gmul);
    zl  := ClippedKaratsubaPolynomialProduct(fl,gl,x,zla,zlb,gmul);
    zm  := ClippedKaratsubaPolynomialProduct(fm,gm,x,zma,zmb,gmul)
           - zh - zl;

    # Combine and clip
    return clip(zh*x^p + zm*x^(p/2) + zl, x, a, b)@\verb+                  +@# Note A
end
\end{lstlisting}
\caption{Clipped Karatsuba polynomial product}
\label{fig:clipped_karatsuba_poly}
\end{figure}
As for the classical case, if polynomials are represented using arrays, the lines annotated with \code{Note A} are performed with array operations.
The parameter \code{gmul} is the coefficient multiplication function and \code{gdegree} gives the degree of the polynomial in whatever representation is used.  
In practice, one would have a cut over to classical multiplication for small sizes rather than recursing all the way down to single coefficients.

\input{tbl_numops}
Table~\ref{tbl:numops} shows how many coefficient multiplications are required by the above implementation to multiply the two polynomials
\begin{align*}
    f =& \phantom{\,+\,}
 \mtwo{-89 & 56  \\ -96 & 72  } x^{15}
+\mtwo{-8  & 64  \\ -32 & 61  } x^{14}
+\mtwo{45  & 66  \\ 69  & 76  } x^{12}
\\&
+\mtwo{-96 & 47  \\ 15  & -85 } x^{11}
+\mtwo{-96 & 62  \\ -74 & -65 } x^{10}
+\mtwo{-92 & 54  \\ -18 & -64 } x^9
+\mtwo{-56 & -28 \\ 56  & -8  } x^8
\\&
+\mtwo{23  & -31 \\ -85 & 94  } x^7
+\mtwo{-45 & -58 \\ 73  & -70 } x^6
+\mtwo{-6  & -7  \\ 72  & 4   } x^5
+\mtwo{-64 & 61  \\ 10  & 45  } x^4
\\&
+\mtwo{-29 & -43 \\ -95 & 16  } x^3
+\mtwo{31  & -9  \\ -42 & 28  } x^2
+\mtwo{-52 & -87 \\ -51 & -27 } x
+\mtwo{-48 & -33 \\ -55 & -22 } \\
\end{align*}
\begin{align*}
    g =& \phantom{\,+\,}
 \mtwo{ 1   & 75  \\ 7   & -15 } x^{15}
+\mtwo{ -22 & 43  \\ 85  & 25  } x^{14}
+\mtwo{ -29 & -90 \\ -38 & 3   } x^{13}
+\mtwo{ 39  & -92 \\ 0   & 18  } x^{12}
\\&
+\mtwo{ 56  & 41  \\ -53 & 6   } x^{11}
+\mtwo{ -10 & 53  \\ -8  & 83  } x^{10}
+\mtwo{ 58  & -98 \\ 61  &  1  } x^{9}
+\mtwo{ -28 & 7   \\ 17  & 36  } x^{8}
\\&
+\mtwo{ -64 & 16  \\ -58 & 64  } x^{7}
+\mtwo{ -76 & -66 \\ 83  & 76  } x^{6}
+\mtwo{  6  & 3   \\ 34  & 8   } x^{5}
+\mtwo{ -80 & -71 \\ -15 & 88  } x^{4}
\\&
+\mtwo{ -9  & -83 \\ 77  & 28  } x^{3}
+\mtwo{ 59  & -28 \\ 48  & 94  } x^{2}
+\mtwo{ 40  & -91 \\ -34 & 32  } x
+\mtwo{ 8   & 39  \\ 9   & -28 } 
\end{align*}
We see the cost of the clipped Karatsuba depends on how close the clipping interval is to the center of the product.  If $a \le b \le p$ then, it depends on $b$, and if $p \le a \le b$ it depends on $a$.  If the interval spans the center, then no savings are achieved.

\subsection*{Clipped FFT Polynomial Multiplication}
The basic scheme for FFT-based multiplication of polynomials $f$ and $g$ starts with zero padded coefficient vectors for $f$ and $g$ of dimension $N = 2^{\lceil \log_2 (\gdeg f + \gdeg g + 1) \rceil}$,
\begin{align*}
    v_f &= [f_0, \ldots, f_{\gdeg f}, 0, \ldots 0] \\
    v_g &= [g_0, \ldots, g_{\gdeg g}, 0, \ldots 0].
\end{align*}
Then FFTs in an appropriate field are computed
\begin{align*}
    \tilde v_f &= \mathrm{FFT} (v_f) \\
    \tilde v_g &= \mathrm{FFT} (v_g).
\end{align*}
The FFT of the product is computed as
\[
    \tilde v_{fg} = [(\tilde v_f)_0 \times (\tilde v_g)_0, \ldots,
                  (\tilde v_f)_{N-1} \times (\tilde v_g)_{N-1}],
\]
where $\times$ denotes multiplication in the chosen field. The coefficient vector for the product is then the inverse FFT of $\tilde v_{fg}$.   Asymptotically, an FFT requires fewer than $\tfrac32 N \log_2 N$ coefficient multiplications. This yields a polynomial multiplication cost triple that, from the two forward and one inverse FFT. 

\input{tbl_fftdep}
Whereas with classical and Karatsuba multiplication, computing unneeded terms has an $O(N^p)$ cost, computing extra FFT terms has only a quasi-linear cost.  So in many situations the straightforward methods of Section~\ref{sec:straightforward} may be used.
To compute a clipped multiplication it will sometimes be possible to avoid computing some of the coefficients of $\tilde v_{fg}$.   The butterfly transformation of the Cooley-Tukey algorithm~\cite{cooley-tukey} does not have every coefficient of $\mathrm{FFT}^{-1}(\tilde v)$ depend on every coefficient of $\tilde v$.   For example, for $N=16$ the dependencies are shown in Table~\ref{tbl:fftdep}.
Thus, depending on the clipping range it is possible to avoid computing some coefficients of $\widetilde {fg}$. 

We do not pursue this in detail here since our focus in this paper is on products from general purpose applications --- but our observation shows that even for extreme precision there is scope for special treatment where clipped products are required.

In the situation where $f\times g$ is of size $2N\times N$, if only the middle third of the product is required, then the method of Hanrot, Quercia and Zimmerman~\cite{HanQuerZim} may be applied, saving a factor of 2.
It remains an open question how far this technique can be generalized.
Nonetheless, there will be clipping ranges and argument lengths where it is worthwhile to pad the factors to use this method.

Additionally, the concept of a truncated FFT~\cite{vdHoevenTruncatedFFT} can be useful.  It computes all of the terms, but avoids unnecessary multiplications.  Other work~\cite{vdHoevenLecerfBlockwise} shows how multiplication may be performed when the arguments may be decomposed into blocks.  

\section{Clipped Integer Multiplication}
\label{sec:clippedint}
The methods for clipped polynomial products can be adapted for the computation of clipped integer products by an analysis of the required number of lower guard digits for carries.
We show how this can be done for classical multiplication.

\subsection*{Clipped Classical Integer Multiplication}

The idea here is really simple. To clip the product to digits 
from $a$ to $b$ you form a na\"{\i}ve clipped product from
$a-G$ to $b$   and then only keep digits from $a$ up.
So $G$ is a number of
guard digits. The sole issue is selecting a value for $G$. And
indeed understanding if there is any value of $G$ that guarantees
correct results. This amounts to a need to understand how
far carries can propagate through the process of the addition that
combines partial products. While in general addition there is no
limit to how far carries can propagate, here there are limits,
and they flow from the fact that if $B$ is the base the
largest a coefficient product can be is $(B-1)^2$ and the high half
of that is quite a lot smaller than $B$.

\begin{theorem}
    To compute the $[a..b]$ clipped product of two base-$B$ integers $f$ and $g$ 
    to within 1 unit in position $a$
    using classical multiplication with partial products $f \cdot g_i B^i$
    requires at most $G = \min(a, \lceil \log_B \gprec g \rceil + 1)$ guard digits.
\end{theorem}
\noindent
By this we mean that the $[a..b]$ clipped product requires summing only the columns for $(fg)_i$, for $a-G \le i \le b$, and the columns $[a-G..a-1]$ are used solely to determine a carry.
\begin{proof}
Let $\gprec f = s+1$ and $\gprec g = t+1$, and
let $c_k$ be the carry into the $k$-th column in the full multiplication tableau,
\[
\begin{array}{c@{~~}c@{~~~}c@{~~~}c@{~~}c@{~~~}c@{~~}c@{~~}c}
 c_{s+t+1} & c_{s+t} & c_{s+t-1} & \cdots & c_s & c_{s-1} & \cdots & c_0 = 0\\
  && & & f_s\cdot g_0 & f_{s-1} \cdot g_0 & \cdots & f_0 \cdot g_0 \\
  && & \iddots\phantom{xx} &  &&  & \iddots\phantom{xx}  \\
  && \phantom{xxxxxxx}\iddots &  & & & \phantom{xxxxxxx}\iddots  \\
  && f_s\cdot g_{t-1} & f_{s-1} \cdot g_{t-1} & \cdots & \cdots & f_0 \cdot g_{t-1} \\
&f_s\cdot g_t & f_{s-1} \cdot g_t & \cdots & \cdots & f_0 \cdot g_t
\end{array}
\]
The sum for the $k$-the column is
\[
s_k = c_k + \sum_{i = 0}^k f_i g_{k-i}  
    \le c_k + \gprec g \cdot (B-1)^2.   
\]
The equation 
uses the fact that $u_i = 0$ for $i \notin [0..\gprec u)$, and
the inequality 
takes into account both that $f_i g_j \le (B-1)^2$ and that $g_{k-i} = 0$ when $i > k$ (\textit{i.e.} the empty lower right hand triangle in the tableau).
The carry into the next column is
\[
c_{k+1} = \lfloor s_k / B \rfloor  \le s_k/B
\]
so
\[
B~c_{k+1} \le c_k + \gprec g \cdot (B-1)^2.
\]
If there is a maximum carry, $c_{\max}$, then it must satisfy
\[
B\,c_{\max} \le  c_{\max} + \gprec g \cdot (B-1)^2 
\]
so we have
\[
c_{\max} \le  \gprec g \cdot (B-1).
\]
It follows that at most $\lceil \log_B \gprec g \rceil + 1$ guard digits are required to guarantee the correct value carried into the clipping range.  No guard digits can be taken before position 0,
therefore clipped classical multiplication for integers needs to sum columns $[a-G..b]$, for $G = \min(a,\lceil \log_B \gprec g \rceil + 1)$, and the columns $[a-G..a-1]$ are used only to compute $c_a$.
\blarghQED
\end{proof}
\noindent
Note that in general a carry can have multiple digits.  With modern computers, the base can be chosen to be large (\textit{e.g.} $2^{64}$) so $\lceil \log_B \gprec g \rceil$ will be 1 for all practical problems and 2 guard digits will suffice.  The details are shown in Figure~\ref{fig:clipped_classical_int}.

\begin{figure}[t]
\begin{lstlisting}[language=maple,style=lstFigStyle]
ClippedClassicalIntegerProduct := proc(f0, g0, B, a, b)
    local f, g, pf, pg, nGuard, carry, s, k, t, i, i0;
    
    # Ensure pg <= pf
    f := f0; pf := iprec(f, B);
    g := g0; pg := iprec(g, B);
    if pg > pf then f,g := g,f;  pf,pg := pg,pf fi;

    # Compute number of guard digits
    nGuard := min(a, ceil(log[B](iprec(g, B))) + 1);
    
    # Form column sums
    carry := 0;
    s     := 0;
    for k from a - nGuard to b do
        if   k < pg then i0 := 0
        elif k < pf then i0 := k - pg
        else             i0 := k - pf
        fi;
        
        t := carry;
        for i from i0 to k do t := t + icoeff(f,B,i) * icoeff(g,B,k-i) od;
        
        carry, t := iquo(t, B), irem(t, B);
        if i >= a then s := isetcoeff(s, B, k, t) fi;
    od;
    return clip(s, B, a, b)
end:

# These would be O(1) or O(b-a) array operations in C.
iprec     := (n, B)       -> ceil(log[B](n+1)):
icoeff    := (n, B, i)    -> irem(iquo(n, B^i), B):
isetcoeff := (n, B, i, v) -> n + v*B^i:
clip      := (n, B, a, b) -> irem(iquo(n, B^a), B^(b-a+1)):
\end{lstlisting}
\caption{Clipped classical integer product}
\label{fig:clipped_classical_int}
\end{figure}

\section{Combining Methods}
\label{sec:combining}
Some of the methods we have described require recursive multiplications of parts.  The recursive calls need not use the same multiplication method.   So there will be platform-dependent thresholds to determine when to use which method.   We note some additional considerations below.
\subsection*{Issues in Mulders Multiplication}
Mulders~\cite{mulders} considered short multiplication of power series, \textit{i.e.} keeping just the top $N$ terms of the product of two series each of length $N$. His scheme could save perhaps 20\% to 30\% of the time that would have been spent had the result been generated by using Karatsuba to form a full product and then just discarding the unwanted low terms. 

Full multiplication of the two inputs computes unwanted parts of the result that correspond to half the partial products that classical multiplication would use. Mulders performs a smaller
full multiply that still computes many unwanted terms, and which also leaves some necessary parts of the result incomplete. Henriot and Zimmermann~\cite{HanZim} investigated just what proportion of the full multiplication would be optimal, but for our purposes it will suffice to approximate that as $0.7N$ by $0.7N$. The parts of the result not computed by this will have the form of a couple of additional instanced of the short product problem, so get handled
by recursion. 

When this scheme is used for integer multiplication the issue of carries from a discarded low part of a product arise. As with
the same issue when using simple classical multiplication this can be handled by using some guard digits, and because Mulders works
with and then discards low partial products beyond those used in
a classical scheme a bound that is good enough for classical can be applied here.

The idea of using a fast complete multiplication that computes somewhat more of a product than will be needed and discarding the excess, but then needing to fill in some gaps, can be generalized
beyond $N\times N$ cases and beyond the case where exactly the top (or bottom) half of a product is needed, and the original Mulders overlap fraction can in general guide usage -- however all additional cases could need fresh analysis to find the exact optimal value for that parameter.

\subsection*{Different Methods for Different Ranges}
We now return to a discussion on our motivating problem, computing clipped integer products in the approximation scheme described in~\cite{watt-issac-2023}.   There, at various points, it is desired to compute clipped products for the most significant, least significant and middle parts of an asymmetric multiplication. Depending on the size of the prefix or suffix, we can use clipped classical, Karatsuba or FFT for these.   For the middle third, a more complex scheme may be used, as shown in Figure~\ref{fig:middle-third}, illustrating a product $f \times g$.
(Diagonal lines from upper left to lower right are multiples of $g$ by a term of $f$.
Diagonal lines from upper right to lower left are multiples of $f$ by a term of $g$.)
This is motivated by Mulders but needs to clip at both high and low ends. Use of a $3N$ by $2N$ variant of Toom-Cook~\cite{cook-phd,1963-toom} gives rather close to
the levels of overlap at each and that would be good for Mulders.

A different and somewhat extreme case would be where the number
of output digits (say $h$) required is much smaller than the size of either input. In that case a classical multiplication can clearly deliver a result in something like $h\times N$ but a decomposition of the strip as in Figure~\ref{fig:middle-third-method} turns out on analysis to start to win once $h$ is significantly larger that the threshold at which Karatsuba breaks even for all the tiled sub-products.

\input{fig_middle_third}
\input{fig_middle_third_method}
\clearpage
\section{Further Thoughts and Conclusions}
\label{sec:thoughts}
We have studied the problem of computing a specified portion of integer and polynomial products, giving some algorithms, cost analysis, bounds on carry propagation and examples.  

We have shown a number of methods, and which is least costly depends on the size of the values to be multiplied and on the interval of the product desired.   For example, if a very few digits of the middle of an integer product are desired, then clipped classical multiplication will give the result in time linear in the input with a good constant factor.  On the other hand, if a substantial fraction of the leading terms of a huge polynomial product are desired, then a straightforward bottom clipped FFT product of the reverse polynomials may be best.

In general, one needs a polyalgorithm that reduces to specific schemes in particular regimes.  We believe the boundaries between these regimes, and the remaining gaps, have not been particularly well explored, and the present work is a step toward filling them.

\IfFileExists{IfExistsUseBBL.bbl}{%

\input{main.bbl}
}{%
\bibliography{main}
}

\newpage

\end{document}

%% file: fig_straightforward_clipped.tex
\begin{figure}[t]
\begin{lstlisting}[language=maple,style=lstFigStyle]
DirectClippedProduct := proc(f, g, t, a, b) 
    local p, r;
    p := f * g;
    r := clip(p, t, a, b);
    return r
end;

BottomClippedProduct := proc(f, g, t, b)
    local clipf, clipg, p, r;
    clipf := clip(f, t, 0, min(b, prec(t,f)-1));
    clipg := clip(g, t, 0, min(b, prec(t,g)-1));
    p := clipf * clipg;
    r := clip(p, t, 0, b);
    return r
end;

ClippedProductFromBottom := proc(f, g, t, a, b)
    local rb, r;
    rb := BottomClippedProduct(f, g, t, b);
    r  := clip(rb,  t, a, b);
    return r
end;
\end{lstlisting}
\caption{Straightforward clipped products.}
\label{fig:straightforward_clipped}
\end{figure}

%% file: tbl_numops.tex
\begin{table}[tb]
\caption{Coefficient multiplications count for $f \times g$ by clipped Karatsuba multiplication}
\begin{scriptsize}
\begin{tabular}{r@{~}|@{~}rrrrrrrrrrrrrrrrrrrrrrrrrrrrrrr}
$a \backslash b$
 &~0 & ~1 & ~2 & ~3 & ~4 & ~5 & ~6 & ~7 & ~8 & ~9 & 10 & 11 & 12 & 13 & 14 & 15 & 16 & 17 & 18 & 19 & 20 & 21 & 22 & 23 & 24 & 25 & 26 & 27 & 28 & 29 & 30 \\
 \hline
  0& 1& 3& 5& 9& 11& 15& 19& 27& 29& 33& 37& 45& 49& 56& 64& 80& 80& 80& 80& 80& 80& 80& 80& 80& 80& 80& 80& 80& 80& 80& 80 \\
  1 &  & 3& 5& 9& 11& 15& 19& 27& 29& 33 & 37& 45& 49& 56& 64& 80& 80& 80& 80& 80& 80& 80& 80& 80& 80& 80& 80& 80& 80& 80& 80 \\
  2 &  &   & 5& 9& 11& 15& 19& 27& 29& 33 & 37& 45& 49& 56& 64& 80& 80& 80& 80& 80& 80& 80& 80& 80& 80& 80& 80& 80& 80& 80& 80 \\
  3 &  &   &   & 9& 11& 15& 19& 27& 29& 33 & 37& 45& 49& 56& 64& 80& 80& 80& 80& 80& 80& 80& 80& 80& 80& 80& 80& 80& 80& 80& 80 \\
  4 &  &   &   &   & 11& 15& 19& 27& 29& 33 & 37& 45& 49& 56& 64& 80& 80& 80& 80& 80& 80& 80& 80& 80& 80& 80& 80& 80& 80& 80& 80 \\
  5 &  &   &   &   &   & 15& 19& 27& 29& 33 & 37& 45& 49& 56& 64& 80& 80& 80& 80& 80& 80& 80& 80& 80& 80& 80& 80& 80& 80& 80& 80 \\
  6 &  &   &   &   &   &   & 19& 27& 29& 33 & 37& 45& 49& 56& 64& 80& 80& 80& 80& 80& 80& 80& 80& 80& 80& 80& 80& 80& 80& 80& 80 \\
  7 &  &   &   &   &   &   &   & 27& 29& 33 & 37& 45& 49& 56& 64& 80& 80& 80& 80& 80& 80& 80& 80& 80& 80& 80& 80& 80& 80& 80& 80 \\
  8 &  &   &   &   &   &   &   &   & 29& 33 & 37& 45& 49& 56& 64& 80& 80& 80& 80& 80& 80& 80& 80& 80& 80& 80& 80& 80& 80& 80& 80 \\
  9 &  &   &   &   &   &   &   &   &   & 33 & 37& 45& 49& 56& 64& 80& 80& 80& 80& 80& 80& 80& 80& 80& 80& 80& 80& 80& 80& 80& 80 \\
  10 &  &   &   &   &   &   &   &   &   &   & 37& 45& 49& 56& 64& 80& 80& 80& 80& 80& 80& 80& 80& 80& 80& 80& 80& 80& 80& 80& 80 \\
  11 &  &   &   &   &   &   &   &   &   &   &   & 45& 49& 56& 64& 80& 80& 80& 80& 80& 80& 80& 80& 80& 80& 80& 80& 80& 80& 80& 80 \\
  12 &  &   &   &   &   &   &   &   &   &   &   &   & 49& 56& 64& 80& 80& 80& 80& 80& 80& 80& 80& 80& 80& 80& 80& 80& 80& 80& 80 \\
  13 &  &   &   &   &   &   &   &   &   &   &   &   &   & 56& 64& 80& 80& 80& 80& 80& 80& 80& 80& 80& 80& 80& 80& 80& 80& 80& 80 \\
  14 &  &   &   &   &   &   &   &   &   &   &   &   &   &   & 64& 80& 80& 80& 80& 80& 80& 80& 80& 80& 80& 80& 80& 80& 80& 80& 80 \\
  15 &  &   &   &   &   &   &   &   &   &   &   &   &   &   &   & 80& 80& 80& 80& 80& 80& 80& 80& 80& 80& 80& 80& 80& 80& 80& 80 \\
  16 &  &   &   &   &   &   &   &   &   &   &   &   &   &   &   &   & 64& 64& 64& 64& 64& 64& 64& 64& 64& 64& 64& 64& 64& 64& 64 \\
  17 &  &   &   &   &   &   &   &   &   &   &   &   &   &   &   &   &   & 56& 56& 56& 56& 56& 56& 56& 56& 56& 56& 56& 56& 56& 56 \\
  18 &  &   &   &   &   &   &   &   &   &   &   &   &   &   &   &   &   &   & 48& 48& 48& 48& 48& 48& 48& 48& 48& 48& 48& 48& 48 \\
  19 &  &   &   &   &   &   &   &   &   &   &   &   &   &   &   &   &   &   &   & 44& 44& 44& 44& 44& 44& 44& 44& 44& 44& 44& 44 \\
  20 &  &   &   &   &   &   &   &   &   &   &   &   &   &   &   &   &   &   &   &   & 36& 36& 36& 36& 36& 36& 36& 36& 36& 36& 36 \\
  21 &  &   &   &   &   &   &   &   &   &   &   &   &   &   &   &   &   &   &   &   &   & 32& 32& 32& 32& 32& 32& 32& 32& 32& 32 \\
  22 &  &   &   &   &   &   &   &   &   &   &   &   &   &   &   &   &   &   &   &   &   &   & 28& 28& 28& 28& 28& 28& 28& 28& 28 \\
  23 &  &   &   &   &   &   &   &   &   &   &   &   &   &   &   &   &   &   &   &   &   &   &   & 26& 26& 26& 26& 26& 26& 26& 26 \\
  24 &  &   &   &   &   &   &   &   &   &   &   &   &   &   &   &   &   &   &   &   &   &   &   &   & 18& 18& 18& 18& 18& 18& 18 \\
  25 &  &   &   &   &   &   &   &   &   &   &   &   &   &   &   &   &   &   &   &   &   &   &   &   &   & 14& 14& 14& 14& 14& 14 \\
  26 &  &   &   &   &   &   &   &   &   &   &   &   &   &   &   &   &   &   &   &   &   &   &   &   &   &   & 10& 10& 10& 10& 10 \\
  27 &  &   &   &   &   &   &   &   &   &   &   &   &   &   &   &   &   &   &   &   &   &   &   &   &   &   &   & 8& 8& 8& 8 \\
  28 &  &   &   &   &   &   &   &   &   &   &   &   &   &   &   &   &   &   &   &   &   &   &   &   &   &   &   &   & 4& 4& 4 \\
  29 &  &   &   &   &   &   &   &   &   &   &   &   &   &   &   &   &   &   &   &   &   &   &   &   &   &   &   &   &   & 3& 3 \\
  30 &  &   &   &   &   &   &   &   &   &   &   &   &   &   &   &   &   &   &   &   &   &   &   &   &   &   &   &   &   &   & 1
\end{tabular}
\end{scriptsize}
\label{tbl:numops}
\end{table}

%% file: tbl_fftdep.tex
\begin{table}
\caption{Coefficient dependencies for inverse FFT on 16 elements}
\begin{center}
    \begin{tabular}{c@{~~~~~}c}
    \textbf{\textit{i}}             & $\mathbf{FFT^{-1}}(\tilde{\text{\textbf{\textit v}}}_\text{\textbf{\textit i}})$  \textbf{dependencies} \\
    \hline
    0    &  $\tilde v_0$ \\
    1    &  $\tilde v_0$, $\tilde v_1$ \\
    2    &  $\tilde v_0$, $\tilde v_2$ \\
    3    &  $\tilde v_1$, $\tilde v_2$, $\tilde v_3$ \\
    4    &  $\tilde v_0$, $\tilde v_4$ \\
    5    &  $\tilde v_1$, $\tilde v_4$, $\tilde v_5$ \\
    6    &  $\tilde v_2$, $\tilde v_4$, $\tilde v_6$ \\
    7    &  $\tilde v_3$, $\tilde v_5$, $\tilde v_6$, $\tilde v_7$ \\
    8    &  $\tilde v_0$, $\tilde v_8$ \\
    9    &  $\tilde v_1$, $\tilde v_8$, $\tilde v_9$ \\
    10 &  $\tilde v_2$, $\tilde v_8$, $\tilde v_{10}$ \\
    11 &  $\tilde v_3$, $\tilde v_9$, $\tilde v_{10}$, $\tilde v_{11}$ \\
    12 &  $\tilde v_4$, $\tilde v_8$, $\tilde v_{12}$ \\
    13 &  $\tilde v_5$, $\tilde v_9$, $\tilde v_{12}$, $\tilde v_{13}$ \\
    14 &  $\tilde v_6$, $\tilde v_{10}$, $\tilde v_{12}$, $\tilde v_{14}$ \\
    15 &  $\tilde v_7$, $\tilde v_{11}$, $\tilde v_{13}$, $\tilde v_{14}$, $\tilde v_{15}$
    \end{tabular}
\end{center}
\label{tbl:fftdep}
\end{table}

%% file: fig_middle_third.tex
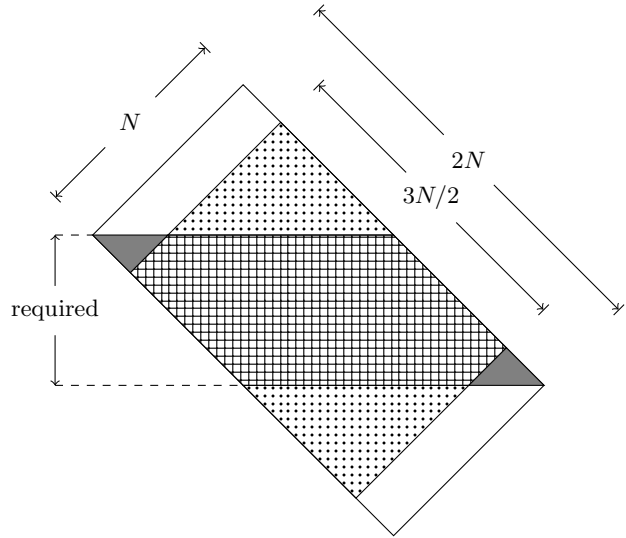
\begin{figure}
    \centering
    \begin{tikzpicture}
       \draw (0,0) -- (2,2) -- (6,-2) -- (4,-4) -- cycle;   
       \draw[dashed] (-.5,0) -- (0,0);
       \draw[dashed] (-.5,-2) -- (2,-2);
       \draw[->] (-.5,-.677) --(-.5,0);
       \node[] at (-.5,-1) {required};
       \draw[->] (-.5,-1.333) --(-.5,-2);
       \draw[pattern=grid] (0.5,-0.5) -- (1,0) -- (4,0)-- (5.5,-1.5) -- (5,-2) -- (2,-2) -- cycle;
       \draw[pattern=dots] (1,0) -- (2.5,1.5) -- (4,0) -- cycle;
       \draw[pattern=dots] (5,-2) -- (3.5,-3.5) -- (2,-2) -- cycle;
       \draw[fill=gray] (0,0) -- (1,0) -- (0.5,-0.5) -- cycle;
       \draw[fill=gray] (5.5,-1.5) -- (6,-2) -- (5,-2) -- cycle;
       \draw [->|] (.167,1.167) -- (-.5,.5);
       \node[] at (.5,1.5) {$N$};
       \draw [->|] (.833,1.833) -- (1.5,2.5);
       \draw [->|] (4.667,1.333) -- (3,3);
       \node[] at (5,1) {$2N$};
       \draw [->|] (5.333,.667) -- (7,-1);
       \draw [->|] (4.167,0.833) -- (3,2);
       \node[] at (4.5,.5) {$3N/2$};
       \draw [->|] (4.833,0.167) -- (6,-1);
    \end{tikzpicture}
    \caption{The desired region of the asymmetric $N \times 2N$ product.
    Horizontal lines are terms of equal degree,
    and the dashed lines show the range of terms desired.
    The $N \times 3N/2$ rectangle is calculated and the cross hatched area is used and the dotted area is not used.   The solid gray areas are calculated separately.}
    \label{fig:middle-third}
\end{figure}

%% file: fig_middle_third_method.tex
\begin{figure}
    \centering
     \begin{tikzpicture}[scale=0.75]
       \draw (0,0) -- (5,5) -- (10,0) -- (5,-5) -- cycle;   
       \draw[dashed] (0,-.9) -- (10,-.9);
       \draw[dashed] (0,-2.1) -- (10,-2.1);
       \draw[->]     (0,-1.3) -- (0,-.9);
       \node[] at (0,-1.5) {required};
       \draw[->]     (0,-1.7) -- (0,-2.1);
       \draw (1.5,-1.5) -- (2.5,-.5) -- (3.5,-1.5) -- (2.5,-2.5) -- cycle;
       \draw[shift={(1.25,0)}] (1.5,-1.5) -- (2.5,-.5) -- (3.5,-1.5) -- (2.5,-2.5) -- cycle;
       \draw[shift={(2.5,0)}] (1.5,-1.5) -- (2.5,-.5) -- (3.5,-1.5) -- (2.5,-2.5) -- cycle;
       \draw[shift={(3.75,0)}] (1.5,-1.5) -- (2.5,-.5) -- (3.5,-1.5) -- (2.5,-2.5) -- cycle;
       \draw[shift={(5,0)}] (1.5,-1.5) -- (2.5,-.5) -- (3.5,-1.5) -- (2.5,-2.5) -- cycle;
       \draw[fill=gray] (.9,-.9) -- (2.1,-.9) -- (1.5,-1.5) -- cycle;
       \draw[fill=gray,shift={(7,0)}] (.9,-.9) -- (2.1,-.9) -- (1.5,-1.5) -- cycle;
       \draw[fill=gray] (2.9,-.9) -- (3.125,-1.125) -- (3.35, -.9) -- cycle;
       \draw[fill=gray] (2.9,-2.1) -- (3.125,-1.875) -- (3.35, -2.1) -- cycle;
       \draw[fill=gray,shift={(1.25,0)}] (2.9,-.9) -- (3.125,-1.125) -- (3.35, -.9) -- cycle;
       \draw[fill=gray,shift={(1.25,0)}] (2.9,-2.1) -- (3.125,-1.875) -- (3.35, -2.1) -- cycle;
       \draw[fill=gray,shift={(2.5,0)}] (2.9,-.9) -- (3.125,-1.125) -- (3.35, -.9) -- cycle;
       \draw[fill=gray,shift={(2.5,0)}] (2.9,-2.1) -- (3.125,-1.875) -- (3.35, -2.1) -- cycle;
       \draw[fill=gray,shift={(3.75,0)}] (2.9,-.9) -- (3.125,-1.125) -- (3.35, -.9) -- cycle;
       \draw[fill=gray,shift={(3.75,0)}] (2.9,-2.1) -- (3.125,-1.875) -- (3.35, -2.1) -- cycle;
       \draw[pattern=grid] (2.75,-1.5)--(3.125,-1.125)--(3.5,-1.5) --(3.125,-1.875) -- cycle;
       \draw[pattern=grid,shift={(1.25,0)}] (2.75,-1.5)--(3.125,-1.125)--(3.5,-1.5) --(3.125,-1.875) -- cycle;
       \draw[pattern=grid,shift={(2.5,0)}] (2.75,-1.5)--(3.125,-1.125)--(3.5,-1.5) --(3.125,-1.875) -- cycle;
       \draw[pattern=grid,shift={(3.75,0)}] (2.75,-1.5)--(3.125,-1.125)--(3.5,-1.5) --(3.125,-1.875) -- cycle;
    \end{tikzpicture}
    \caption{A scheme to compute the middle digits.  The squares are calculated by a size-appropriate method (classical, Karatsuba, FFT).  The hatched areas of overlap are double counted, so must be subtracted.  The solid areas are calculated sparately.}
    \label{fig:middle-third-method}
    \begin{tikzpicture}
        
    \end{tikzpicture}
\end{figure}
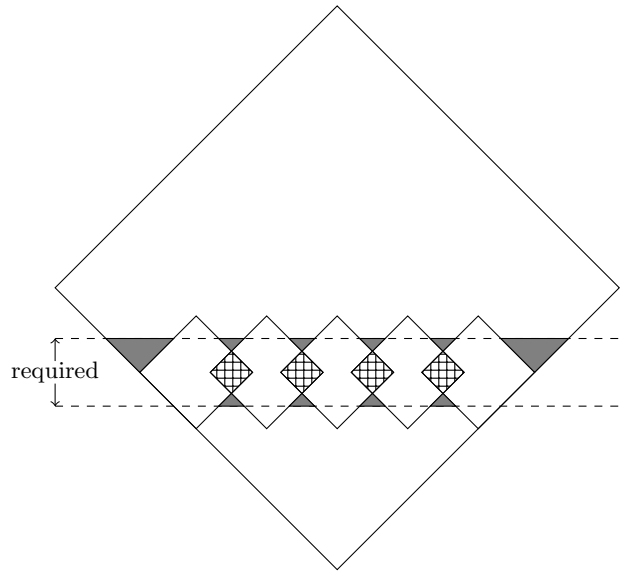

%% file: main.bbl

\begin{thebibliography}{10}
\providecommand{\url}[1]{\texttt{#1}}
\providecommand{\urlprefix}{URL }
\providecommand{\doi}[1]{https://doi.org/#1}

\bibitem{cook-phd}
Cook, S.A.: On the Minimum Computation Time of Functions. Ph.D. thesis, Harvard University (1966)

\bibitem{cooley-tukey}
Cooley, J.W., Tukey, J.W.: An algorithm for the machine calculation of complex {F}ourier series. Mathematics of Computation  \textbf{19}(90),  297--301 (1965)

\bibitem{HanQuerZim}
Hanrot, G., Quercia, M., Zimmermann, P.: The middle product algorithm {I}: Speeding up the division and square root of power series. Applicable Algebra in Engineering, Communication and Computing  \textbf{14},  415--438 (2004)

\bibitem{HanZim}
Hanrot, G., Zimmermann, P.: A long note on {M}ulders' short product. J. Symbolic Computation  \textbf{37},  391--401 (2004)

\bibitem{vdHoevenTruncatedFFT}
van~der Hoeven, J.: The truncated {F}ourier transform and applications. In: Proc. ISSAC 2004. pp. 290--–296. ACM, New York (2004)

\bibitem{vdHoevenLecerfBlockwise}
van~der Hoeven, J., Lecerf, G.: On the complexity of multivariate blockwise polynomial multiplication. In: Proc. ISSAC 2012. pp. 211--–218. ACM, New York (2012)

\bibitem{karatsuba}
Karatsuba, A., Yu., O.: Multiplication of many-digital numbers by automatic computers. Proceedings of the USSR Academy of Sciences  \textbf{145},  293--294 (1962), translation in the academic journal Physics-Doklady, 7 (1963), pp. 595--596

\bibitem{knuth-vol4b}
Knuth, D.E.: The Art of Computer Programming, Volume 4b: Combinatorial Algorithms, Part 2. Addison-Wesley, Boston (2022)

\bibitem{mulders}
Mulders, T.: On short multiplication and division. In: Proceedings AAECC 11, 1. pp. 69--88 (2000)

\bibitem{2024-norman-watt}
Norman, A.C., Watt, S.M.: A symbolic computing perspective on software systems  \textbf{arxiv:2406.09085},  1--18 (2024)

\bibitem{1963-toom}
Toom, A.L.: {The complexity of a scheme of functional elements realizing the multiplication of integers}. Soviet Mathematics Doklady  \textbf{3},  714--716 (1963)

\bibitem{watt-issac-2023}
Watt, S.M.: Efficient generic quotients using exact arithmetic. In: Proc. International Symposium on Symbolic and Algebraic Computation ({ISSAC 2023}). pp. 535--544. ACM, New York (2023)

\bibitem{watt-casc-2023}
Watt, S.M.: Efficient quotients of non-commutative polynomials. In: $25^{th}$ International Workshop on Computer Algebra in Scientific Computing ({CASC} 2023). pp. 370--392. Springer Cham LNCS 14139, New York (2023)

\end{thebibliography}
